\documentclass[draftcls,onecolumn,12pt]{IEEEtran}

\usepackage [latin1]{inputenc}
\usepackage[usenames,dvipsnames]{xcolor}
\usepackage[tbtags]{amsmath}
\usepackage{amsthm}
\usepackage{subfigure}
\usepackage{graphicx}
\usepackage{epstopdf}
\usepackage{amssymb}
\usepackage{caption}
\usepackage{threeparttable}
\usepackage{color}

\newtheorem{theorem}{Theorem}

\newtheorem{lemma}{Lemma}

\newtheorem{definition}{Definition}

\usepackage[numbers,sort&compress]{natbib}

\begin{document}
\title{A Complete Proof of an Important Theorem for Variable-to-Variable Length Codes}

\author{Wei Yan,~\IEEEmembership{Member,~IEEE}
        and Yunghsiang S. Han,~\IEEEmembership{Fellow,~IEEE}
\thanks{Y. S. Han is with Shenzhen Institute for Advanced Study, University of Electronic Science and Technology of China, Shenzhen, China, email: yunghsiangh@gmail.com.}
}
\maketitle
\begin{abstract}
Variable-to-variable length~(VV) codes are a class of lossless source coding.
As their name implies, VV codes encode a variable-length sequence of source symbols into a variable-length codeword.
This paper will give a complete proof of an important theorem for variable-to-variable length codes.
\end{abstract}
\IEEEpeerreviewmaketitle

\section{Preliminaries}\label{pre}

VV codes can be considered a concatenation of variable-to-fixed length~(VF) codes and fixed-to-variable length (FV) codes. 

\subsection{Variable-to-fixed length codes}
A VF encoder comprises a \emph{parser} and a \emph{string encoder}.
First, the parser partitions the source sequence $x$ into a concatenation of phrases $x^{1}$,$x^{2}$,$\cdots$ from a dictionary $\mathcal{D}$, that is, $x^{i}\in\mathcal{D}$.
Next, the string encoder maps the phrase $x^{i}\in\mathcal{D}$ into the fixed-length string.
To ensure the completeness and uniqueness of the segmentation of the source sequence, $\mathcal{D}$ must be proper and complete.
\begin{definition}[\cite{FK72}]
\begin{enumerate}
\item If every variable-length string $\alpha_{i}\in\mathcal{D}$ is not a prefix of another variable-length string $\alpha_{j}\in\mathcal{D}$, then $\mathcal{D}$ is termed proper.
\item If every infinite sequence has a prefix in $\mathcal{D}$, then $\mathcal{D}$ is termed complete.
\end{enumerate}
\end{definition}
{For example, the dictionary $\mathcal{D}=\{0,10,11\}$ over $\{0,1\}$ is clearly proper.
If the first element of the infinite sequence is $0$, then $0\in\mathcal{D}$ is its prefix;
if the first element of the infinite sequence is $1$, then $10\in\mathcal{D}$ or $11\in\mathcal{D}$ is its prefix.
Therefore, the dictionary $\mathcal{D}$ is complete.}

\subsection{Variable-to-variable length codes}
VV codes can be considered a concatenation of VF and FV codes. First, the VF encoder maps the variable-length string $\alpha\in\mathcal{D}$ into the fixed-length string, and then the FV encoder maps the fixed-length string into the variable-length string.
Nishiara \emph{et al.}~\cite{VV00} define the almost surely complete~(ASC) dictionary.
\begin{definition}[\cite{VV00}]
For every infinite sequence, if the probability there exists a string in dictionary $\mathcal{D}$ is a prefix of the infinite sequence is $1$, then $\mathcal{D}$ is termed almost surely complete.
\end{definition}
An example of a dictionary $\mathcal{D}$ over $\{0,1\}$ that proper and ASC is $\mathcal{D}=\{0,10,110,1110,\cdots\}$. However, it is not complete because the all-one infinite sequence has no prefix in $\mathcal{D}$.

\section{The important theorem for VV codes}\label{sec_lb}
\begin{theorem}~\cite{VV00} \label{lemma1}
Let $S=(P,\mathcal{A})$ denote a discrete memoryless source with entropy $H(P)<\infty$ and a countable alphabet $\mathcal{A}$.
Given a VV code $\mathcal{C}$ with a proper and ASC dictionary $\mathcal{D}$, then
\begin{equation*}
H(\mathcal{D})= H(P)\overline{l(\mathcal{D})},
\end{equation*}
where $H(\mathcal{D})=-\sum_{\alpha\in\mathcal{D}}P(\alpha)\log_{2}P(\alpha)$ denotes the entropy of $\mathcal{D}$
and $\overline{l(\mathcal{D})}=\sum_{\alpha\in\mathcal{D}}P(\alpha)|\alpha|$ denotes the average length of $\mathcal{D}$.
\end{theorem}
Theorem~\ref{lemma1} was first introduced by Nishiara \emph{et al.}~\cite{VV00}, but they did not give complete proof.
The proof for the proper and complete dictionary and the finite alphabet can be found in~\cite{FK72}.
When studying the entropy of randomly stopped sequences,
Ekroot \emph{et al.}~\cite{LT91} gave the proof of the proper and ASC dictionary and the finite alphabet version of Theorem~\ref{lemma1}.
In addition, a similar theorem, called \emph{conservation of entropy}~\cite{S99}, is for memory sources. Therefore, we here give the first complete proof of Theorem~\ref{lemma1} in Section~\ref{A1}.
Unlike previous work, the contribution of this proof is the first to give a proof of the infinite countable alphabet version of Theorem~\ref{lemma1}.

\section{Proof of Theorem~\ref{lemma1}} \label{A1}
The proof of Theorem~\ref{lemma1} in this Section draws on the proof techniques in~\cite{FK72,LT91}.
Some parts of the proof are borrowed from~\cite{FK72,LT91}. For a better understanding of the reader, we will give complete proof.
Let us first give some definitions.
\begin{definition}
Suppose $\mathcal{D}$ is a  dictionary.
\begin{enumerate}
\item
Let $\mathcal{A}^{n}$ denote all strings of length $n$ over the alphabet $\mathcal{A}$.
For any positive integer $n$, the three dictionaries are defined as follows:
\begin{equation*}
	\begin{aligned}
			T_n&\triangleq  \{\alpha\in\mathcal{A}^{n} \,\Big| \, \hbox{$\beta$ is not a prefix of $\alpha$, for $\forall$ $\beta\in\mathcal{D}$  }  \} ,  \\
		\mathcal{D}_{n}^{\perp}&\triangleq \{ \alpha\in\mathcal{D} \,\Big| \, |\alpha|=n \}  \cup T_n,  \\
		\mathcal{D}_{n}&\triangleq \{ \alpha\in\mathcal{D} \,\Big| \, |\alpha|<n \}\cup \mathcal{D}_{n}^{\perp}.
	\end{aligned}
\end{equation*}
In particular, $T_1=\{\alpha\in\mathcal{A} \,|\,\alpha\notin \mathcal{D} \} $ and $\mathcal{D}_{1}^{\perp}=\mathcal{D}_{1}=\mathcal{A}$.
\item
For every string $\beta$, let
\[
(\mathcal{D},\beta)   \triangleq\{\alpha\in\mathcal{D} \,\Big| \, \hbox{$\beta$ is a prefix of $\alpha$}  \} .
\]
\item For every $\alpha\in\mathcal{D}$, let $\mathcal{D}[\alpha]$ denote a dictionary as follows:
\[
   \mathcal{D}[\alpha]\triangleq (\mathcal{D}\setminus\{\alpha\})\cup\alpha\mathcal{A},
\]
where $\alpha\mathcal{A}\triangleq\{\alpha\beta\mid\beta\in\mathcal{A}\}$.
The dictionary $\mathcal{D}[\alpha]$ is said to be an extension of dictionary $\mathcal{D}$, and $\alpha$ is termed the extending string from $\mathcal{D}$ to $\mathcal{D}[\alpha]$.
\end{enumerate}
\end{definition}
When $\mathcal{D}$ is proper and complete, then $\mathcal{D}[\alpha]$ is also proper and complete.
Before proving Theorem~\ref{lemma1}, an auxiliary lemma is introduced.
\begin{lemma}  \label{lemma9}
\begin{enumerate}	
\item[(1)] If the dictionary $\mathcal{D}$ is proper, then $\mathcal{D}_{n}$ is proper and complete.
\item[(2)] $T_m$ is the set of extending strings from $\mathcal{D}_{m}$ to $\mathcal{D}_{m+1}$.
\item[(3)] If the dictionary $\mathcal{D}$ is proper and ASC, then
\[
\sum_{\alpha\in(\mathcal{D},\beta)}P(\alpha)=P(\beta),
\]
for any given string $\beta$, where $\beta$ has no prefix less than length $|\beta|$ in $\mathcal{D}$.
\item[(4)] If the dictionary $\mathcal{D}$ is proper, then
\[
\sum_{\beta\in\mathcal{D}_{m}^{\perp}}(\mathcal{D},\beta)=\{\alpha\in\mathcal{D} \,\Big|\,|\alpha|\geq m \} ,
\]
for all $m\in\mathbb{N}^{+}$.
\end{enumerate}
\end{lemma}
\begin{proof}
\begin{enumerate}	
\item[(1)] First, we prove that $\mathcal{D}_{n}$ is proper.
	The dictionary $\mathcal{D}_{n}$ can be written as follows:
\[
  	\mathcal{D}_{n}\triangleq \{ \alpha\in\mathcal{D} \,\Big| \, |\alpha|\leq n \}\cup T_{n}.
\]
The following two cases are considered for any string $\alpha_i \in \mathcal{D}_{n}$.
\begin{enumerate}
\item Assume $\alpha_i \in \{ \alpha\in\mathcal{D} \,| \, |\alpha|\leq n \}$. Since $\mathcal{D}$ is proper,
$\alpha_i$ is not a prefix of another string $\alpha_j \in \{ \alpha\in\mathcal{D} \,| \, |\alpha|\leq n \}$.
Due to the definition of $T_{n}$, $\alpha_i$ is not a prefix of $\alpha_j$ for all $\alpha_j \in T_{n}$.
\item Assume $\alpha_i \in T_{n}$. Due to $|\alpha_i|=n$, $\alpha_i$ is not a prefix of $\alpha_j$ for all $\alpha_j \in\mathcal{D}_{n}\setminus\{\alpha_i\}$.
\end{enumerate}	
Therefore, $\mathcal{D}_{n}$ is proper.

Second, we prove that $\mathcal{D}_{n}$ is complete.
For any infinite sequence, we consider its first $n$-bit string $\gamma$.
Assume there exists $\alpha \in \mathcal{D}$ such that $\alpha$ is a prefix of $\gamma$.
Then, $\alpha \in \{ \alpha\in\mathcal{D} \,| \, |\alpha|\leq n \} \subseteq \mathcal{D}_{n}$ is the prefix of the infinite sequence.
Assume $\beta$ is not a prefix of $\gamma$ for all $\beta\in\mathcal{D}$.
Then, $\gamma\in T_n\subseteq \mathcal{D}_{n}$ is the prefix of the infinite sequence.
This part of the proof is complete.
\item [(2)]
According to the definition of $\mathcal{D}_{n}$,
the set of extending strings from $\mathcal{D}_{m}$ to $\mathcal{D}_{m+1}$
is essentially composed of elements of length $m$ in $\mathcal{D}_{m}$ that do not belong to $\mathcal{D}$.
Because $T_m$ happens to consist of elements of length $m$ in $\mathcal{D}_{m}$ that do not belong to $\mathcal{D}$.
Therefore, $T_m$ is the set of extending strings from $\mathcal{D}_{m}$ to $\mathcal{D}_{m+1}$.
\item[(3)] Consider all infinite sequences with probabilities greater than $0$ that begin with $\beta$, denote it as $\mathcal{H}$.
First, the sum of the probabilities of these sequences is $P(\beta)$.
Second, since $\mathcal{D}$ is proper and ASC, the infinite sequences considered have unique prefixes in $\mathcal{D}$.
Because $\beta$ has no prefix less than length $|\beta|$ in $\mathcal{D}$,
the length of prefix $\alpha\in\mathcal{D}$ of an infinite sequence starting with $\beta$ is greater than or equal to $|\beta|$. Thus, $\beta$ is a prefix of $\alpha\in\mathcal{D}$.
$\mathcal{H}$ has a subset of infinite sequences starting with $\alpha$,
the sequences in this set have the unique prefix $\alpha\in\mathcal{D}$,
and the sum of the probabilities of these sequences is $P(\alpha)$.
Then, we obtain
\[
  P(\beta)=\sum_{\gamma\in\mathcal{H}}P(\gamma)=\sum_{\alpha\in(\mathcal{D},\beta)}P(\alpha).
\]
\item[(4)] First, due to $\beta\in\mathcal{D}_{m}^{\perp}$, we have $|\beta|=m$.
Thus, $|\alpha|\geq m$ for every $\alpha\in(\mathcal{D},\beta)$.
Therefore, we obtain
\[
\sum_{\beta\in\mathcal{D}_{m}^{\perp}}(\mathcal{D},\beta)\subseteq \{\alpha\in\mathcal{D} \,\Big|\,|\alpha|\geq m \} .
\]
Second, due to Lemma~\ref{lemma9}$(1)$, $\mathcal{D}_{m}$ is proper and complete.
For every $\gamma\in\{\alpha\in\mathcal{D} \,|\,|\alpha|\geq m \}$,
$\gamma$ has a unique prefix $\beta\in\mathcal{D}_{m}$.
Since $\gamma$ belongs to $\mathcal{D}$, $\beta\notin \{ \alpha\in\mathcal{D} \,| \, |\alpha|<m \} $;
that is, $\beta\in\mathcal{D}_{m}^{\perp}$.
Thus, we have
$$\gamma\in\sum_{\beta\in\mathcal{D}_{m}^{\perp}}(\mathcal{D},\beta). $$
Therefore, we obtain
\[
\sum_{\beta\in\mathcal{D}_{m}^{\perp}}(\mathcal{D},\beta)\supseteq \{\alpha\in\mathcal{D} \,\Big|\,|\alpha|\geq m \} .
\]
\end{enumerate}	
\end{proof}

Now, we begin the proof of the theorem.
\begin{theorem}[Theorem~\ref{lemma1} Restated]
	Let $S=(P,\mathcal{A})$ denote a discrete memoryless source with entropy $H(P)<\infty$ and a countable alphabet $\mathcal{A}$.
	Given a VV code $\mathcal{C}$ with a proper and ASC dictionary $\mathcal{D}$, then
	\begin{equation}\label{e11}
		H(\mathcal{D})= H(P)\overline{l(\mathcal{D})},
	\end{equation}
	where $H(\mathcal{D})=-\sum_{\alpha\in\mathcal{D}}P(\alpha)\log_{2}P(\alpha)$ and $\overline{l(\mathcal{D})}=\sum_{\alpha\in\mathcal{D}}P(\alpha)|\alpha|$.
\end{theorem}
\begin{proof}
The proof is divided into three parts.
First, we prove that if the dictionary satisfies  \eqref{e11}, then the dictionary after finite extensions also satisfies  \eqref{e11}.
Next, the following equation will be proved:
\begin{equation}\label{e12}
	H(\mathcal{D}_{n})= H(P)\overline{l(\mathcal{D}_{n})},
\end{equation}
for all $n\in\mathbb{N}^{+}$.
Finally, the proof for \eqref{e11} is presented.
\begin{enumerate}
	\item Suppose $\mathcal{S}$ is a  dictionary.
	We  need to prove that when $\mathcal{S}$ satisfies  \eqref{e11}, then $\mathcal{S}[\alpha]$ after one extension also satisfies \eqref{e11}.
	Note that
	\begin{equation*}
		\begin{aligned}
			\overline{l(\mathcal{S}[\alpha])}&= \sum_{\beta\in\mathcal{S}\setminus\{\alpha\}}P(\beta)|\beta|+ \sum_{\beta\in\alpha\mathcal{A}}P(\beta)|\beta|                             \\
			&= \sum_{\beta\in\mathcal{S}}P(\beta)|\beta|-P(\alpha)|\alpha|+ P(\alpha)(|\alpha|+1)     \\
			&= \overline{l(\mathcal{S})}+P(\alpha),
		\end{aligned}
	\end{equation*}
	and
	\begin{equation*}
		\begin{aligned}
			H(\mathcal{S}[\alpha])&= -\sum_{\beta\in\mathcal{S}\setminus\{\alpha\}}P(\beta)\log_{2}P(\beta)-\sum_{\beta\in\alpha\mathcal{A}}P(\beta)\log_{2}P(\beta)                                \\
			&= -\sum_{\beta\in\mathcal{S}}P(\beta)\log_{2}P(\beta)+P(\alpha)\log_{2}P(\alpha)-\sum_{\beta\in\mathcal{A}}P(\alpha)P(\beta)\log_{2}P(\alpha)P(\beta)      \\
			&= H(\mathcal{S})+P(\alpha)\log_{2}P(\alpha)-P(\alpha)\log_{2}P(\alpha)-P(\alpha)\sum_{\beta\in\mathcal{A}}P(\beta)\log_{2}P(\beta)      \\
			&= H(P)\overline{l(\mathcal{S})}+H(P)P(\alpha)    \\
			&= H(P)\overline{l(\mathcal{S}[\alpha])}.
		\end{aligned}
	\end{equation*}
	We have proved that, for one extension, the extended dictionary also satisfies \eqref{e11}. By a similar process, we can prove that, for any finite number of extensions, the extended dictionary satisfies \eqref{e11}. The first part of the proof is complete.
	\item
	We prove \eqref{e12} by mathematical induction.
	When $n=1$, we have $\mathcal{D}_{1}=\mathcal{A}$ and
	\[
	H(P)\overline{l(\mathcal{D}_{1})}=H(P)\times1=H(\mathcal{D}_{1}).
	\]
	Suppose  \eqref{e12} holds when $n=m$.
	Now, we consider the extension process from $\mathcal{D}_{m}$ to $\mathcal{D}_{m+1}$.
	If $|T_m|<\infty$, then $\mathcal{D}_{m+1}$ is obtained by $\mathcal{D}_{m}$ after finite number of extensions.
	We obtain $H(\mathcal{D}_{m+1})= H(P)\overline{l(\mathcal{D}_{m+1})}$ due to the first part of the proof.
	If $|T_m|=\infty$, because $\mathcal{A}$ is countable and the length of the elements in $T_m$ are all $m$, $T_m$ is also countable.
	Therefore, suppose $T_m\triangleq\{\alpha_{i}\}_{i=1}^{\infty}$, the extension process from $\mathcal{D}_{m}$ to $\mathcal{D}_{m+1}$ is as follows:
	\begin{equation*}
		\begin{aligned}
			\mathcal{D}_{m+1,1}\triangleq &(\mathcal{D}_{m}\setminus\{\alpha_{1}\})\cup \alpha_{1}\mathcal{A},                     \\
			\mathcal{D}_{m+1,2}\triangleq & (\mathcal{D}_{m+1,1}\setminus\{\alpha_{2}\})\cup \alpha_{2}\mathcal{A}= (\mathcal{D}_{m}\setminus\{\alpha_{i}\}_{i=1}^{2})\cup \{\alpha_{i}\mathcal{A}\}_{i=1}^{2},     \\
			\vdots &  \\
			\mathcal{D}_{m+1,k} \triangleq &  (\mathcal{D}_{m+1,k-1}\setminus\{\alpha_{k}\})\cup \alpha_{k}\mathcal{A}= (\mathcal{D}_{m}\setminus\{\alpha_{i}\}_{i=1}^{k})\cup \{\alpha_{i}\mathcal{A}\}_{i=1}^{k},     \\
			\vdots &
		\end{aligned}
	\end{equation*}
	Then, we have the following three equations.
	\begin{enumerate}
		\item[(\romannumeral1)] $H(\mathcal{D}_{m+1,k})= H(P)\overline{l(\mathcal{D}_{m+1,k})}$, for all $k\in\mathbb{N}^{+}$.
		\item[(\romannumeral2)] $\lim\limits_{k\to +\infty}\mathcal{D}_{m+1,k}=\mathcal{D}_{m+1}$.
		\item[(\romannumeral3)] $\mathcal{D}_{m+1}=(\mathcal{D}_{m}\setminus\{\alpha_{i}\}_{i=1}^{\infty})\cup \{\alpha_{i}\mathcal{A}\}_{i=1}^{\infty}$.
	\end{enumerate}
	Next, we prove the following two equations:
	\begin{equation}\label{e13}
		\begin{aligned}
			\lim\limits_{k\to +\infty}\overline{l(\mathcal{D}_{m+1,k})}&=\overline{l(\mathcal{D}_{m+1})},            			\\
			\lim\limits_{k\to +\infty} H(\mathcal{D}_{m+1,k})&=H(\mathcal{D}_{m+1}).
		\end{aligned}
	\end{equation}
	First, we have
	\begin{equation*}
		\begin{aligned}
			\overline{l(\mathcal{D}_{m+1})}& = \sum_{\alpha\in\mathcal{D}_{m+1}}P(\alpha)|\alpha|                               \\
			&\geq \sum_{\alpha\in\mathcal{D}_{m+1,k}}P(\alpha)|\alpha|       \\
			& = \overline{l(\mathcal{D}_{m+1,k})}      \\
			&\geq \sum_{\alpha\in(\mathcal{D}_{m}\setminus\{\alpha_{i}\}_{i=1}^{\infty})\cup\{\alpha_{i}\mathcal{A}\}_{i=1}^{k}}P(\alpha)|\alpha|.
		\end{aligned}
	\end{equation*}
	Taking the limit $k\rightarrow\infty$, we obtain
	\begin{equation*}
		\begin{aligned}
			\overline{l(\mathcal{D}_{m+1})}&\geq \lim\limits_{k\to +\infty} \overline{l(\mathcal{D}_{m+1,k})}     \\
			&\geq \lim\limits_{k\to +\infty} \sum_{\alpha\in(\mathcal{D}_{m}\setminus\{\alpha_{i}\}_{i=1}^{\infty})\cup\{\alpha_{i}\mathcal{A}\}_{i=1}^{k}}P(\alpha)|\alpha|  \\
			&=\overline{l(\mathcal{D}_{m+1})}.
		\end{aligned}
	\end{equation*}
	Then, we have
	\begin{equation*}
		\begin{aligned}
			H(\mathcal{D}_{m+1})&= -\sum_{\alpha\in(\mathcal{D}_{m}\setminus\{\alpha_{i}\}_{i=1}^{\infty})\cup \{\alpha_{i}\mathcal{A}\}_{i=1}^{k}}P(\alpha)\log_{2}P(\alpha)-\sum_{\alpha\in\{\alpha_{i}\mathcal{A}\}_{i=k+1}^{\infty}}P(\alpha)\log_{2}P(\alpha)       \\
			&\overset{(a)}{\geq}-\sum_{\alpha\in(\mathcal{D}_{m}\setminus\{\alpha_{i}\}_{i=1}^{\infty})\cup \{\alpha_{i}\mathcal{A}\}_{i=1}^{k}}P(\alpha)\log_{2}P(\alpha)-\sum_{\alpha\in\{\alpha_{i}\}_{i=k+1}^{\infty}}P(\alpha)\log_{2}P(\alpha)       \\
			&=-\sum_{\alpha\in(\mathcal{D}_{m}\setminus\{\alpha_{i}\}_{i=1}^{k})\cup \{\alpha_{i}\mathcal{A}\}_{i=1}^{k}}P(\alpha)\log_{2}P(\alpha)       \\
			&= H(\mathcal{D}_{m+1,k})    \\
			&\geq -\sum_{\alpha\in(\mathcal{D}_{m}\setminus\{\alpha_{i}\}_{i=1}^{\infty})\cup \{\alpha_{i}\mathcal{A}\}_{i=1}^{k}}P(\alpha)\log_{2}P(\alpha),
		\end{aligned}
	\end{equation*}
	where $(a)$ is due to $-\sum_{\alpha\in\alpha_{i}\mathcal{A}}P(\alpha)\log_{2}P(\alpha)\geq -P(\alpha_{i})\log_{2}P(\alpha_{i})$, for all $i\in\mathbb{N}^{+}$.
	Taking the limit $k\rightarrow\infty$, we obtain
	\begin{equation*}
		\begin{aligned}
			H(\mathcal{D}_{m+1})&\geq \lim\limits_{k\to +\infty} H(\mathcal{D}_{m+1,k})    \\
			&\geq \lim\limits_{k\to +\infty} -\sum_{\alpha\in(\mathcal{D}_{m}\setminus\{\alpha_{i}\}_{i=1}^{\infty})\cup \{\alpha_{i}\mathcal{A}\}_{i=1}^{k}}P(\alpha)\log_{2}P(\alpha) \\
			&=H(\mathcal{D}_{m+1}).
		\end{aligned}
	\end{equation*}
	Equation \eqref{e13} is proved. From the perspective of mathematical analysis, equation \eqref{e13} essentially considers whether the function and the limit can be exchanged.
	For example, $\lim\limits_{k\to +\infty} H(\mathcal{D}_{m+1,k})=H(\lim\limits_{k\to +\infty}\mathcal{D}_{m+1,k})$.
	Finally, form equation \eqref{e13}, we have
	\begin{equation*}
		\begin{aligned}
			H(\mathcal{D}_{m+1})&= \lim\limits_{k\to +\infty} H(\mathcal{D}_{m+1,k})       \\
			&= \lim\limits_{k\to +\infty} H(P)\overline{l(\mathcal{D}_{m+1,k})}   \\
			&= H(P)\overline{l(\mathcal{D}_{m+1})} .
		\end{aligned}
	\end{equation*}
	The second part of the proof is complete.
	 \item
 We prove the following two equations similar to equation \eqref{e13}.
 \begin{equation}\label{e14}
  \begin{aligned}
   \lim\limits_{m\to +\infty}\overline{l(\mathcal{D}_{m})}&=\overline{l(\mathcal{D})}.            \\
   \lim\limits_{m\to +\infty} H(\mathcal{D}_{m})&=H(\mathcal{D}).
  \end{aligned}
 \end{equation}
First, due to Lemma~\ref{lemma9}$(3)$, we obtain
 \[
 \sum_{\alpha\in(\mathcal{D},\beta)}P(\alpha)|\alpha|\geq
 |\beta|\sum_{\alpha\in(\mathcal{D},\beta)}P(\alpha)=P(\beta)|\beta|,
 \]
 for any given $\beta\in\mathcal{D}_{m}^{\perp}$.
Furthermore, from Lemma~\ref{lemma9}$(4)$, we have
 \[
\sum_{\alpha\in\mathcal{D},|\alpha|\geq m}P(\alpha)|\alpha|= \sum_{\beta\in\mathcal{D}_{m}^{\perp}} \sum_{\alpha\in(\mathcal{D},\beta)}P(\alpha)|\alpha|    \geq   \sum_{\beta\in\mathcal{D}_{m}^{\perp}}P(\beta)|\beta| .
 \]
Therefore, we obtain
 \begin{equation*}
  \begin{aligned}
   \overline{l(\mathcal{D})}& = \sum_{\alpha\in\mathcal{D},|\alpha|<m}P(\alpha)|\alpha|+ \sum_{\alpha\in\mathcal{D},|\alpha|\geq m}P(\alpha)|\alpha|                             \\
   &\geq\sum_{\alpha\in\mathcal{D},|\alpha|<m}P(\alpha)|\alpha|+ \sum_{\beta\in\mathcal{D}_{m}^{\perp}}P(\beta)|\beta|      \\
   & = \overline{l(\mathcal{D}_{m})}      \\
   &\geq \sum_{\alpha\in\mathcal{D},|\alpha|<m}P(\alpha)|\alpha|.
  \end{aligned}
 \end{equation*}
 Taking the limit $m\rightarrow\infty$, we obtain
 \begin{equation*}
  \begin{aligned}
   \overline{l(\mathcal{D})}&\geq \lim\limits_{m\to +\infty} \overline{l(\mathcal{D}_{m})}     \\
   &\geq \lim\limits_{m\to +\infty}\sum_{\alpha\in\mathcal{D},|\alpha|<m}P(\alpha)|\alpha|  \\
   &=\overline{l(\mathcal{D})}.
  \end{aligned}
 \end{equation*}
 Then, due to Lemma~\ref{lemma9}$(3)$, we obtain
 \[
 -\sum_{\alpha\in(\mathcal{D},\beta)}P(\alpha)\log_{2}P(\alpha)\geq -\log_{2}P(\beta)
 \sum_{\alpha\in(\mathcal{D},\beta)}P(\alpha)=-P(\beta)\log_{2}P(\beta),
 \]
 for any given $\beta\in\mathcal{D}_{m}^{\perp}$.
 Furthermore, from Lemma~\ref{lemma9}$(4)$, we have
 \[
 -\sum_{\alpha\in\mathcal{D},|\alpha|\geq m}P(\alpha)\log_{2}P(\alpha)=
 -\sum_{\beta\in\mathcal{D}_{m}^{\perp}} \sum_{\alpha\in(\mathcal{D},\beta)}P(\alpha)\log_{2}P(\alpha) \geq   -\sum_{\beta\in\mathcal{D}_{m}^{\perp}}P(\beta)\log_{2}P(\beta)  .
 \]
 Therefore, we obtain
 \begin{equation*}
  \begin{aligned}
   H(\mathcal{D})&= -\sum_{\alpha\in\mathcal{D},|\alpha|<m}P(\alpha)\log_{2}P(\alpha)-\sum_{\alpha\in\mathcal{D},|\alpha|\geq m}P(\alpha)\log_{2}P(\alpha)       \\
 &\geq-\sum_{\alpha\in\mathcal{D},|\alpha|<m}P(\alpha)\log_{2}P(\alpha)-\sum_{\beta\in\mathcal{D}_{m}^{\perp}}P(\beta)\log_{2}P(\beta)        \\
   &= H(\mathcal{D}_{m})    \\
   &\geq -\sum_{\alpha\in\mathcal{D},|\alpha|<m}P(\alpha)\log_{2}P(\alpha),
  \end{aligned}
 \end{equation*}
 Taking the limit $m\rightarrow\infty$, we obtain
 \begin{equation*}
  \begin{aligned}
   H(\mathcal{D})&\geq \lim\limits_{m\to +\infty} H(\mathcal{D}_{m})    \\
   &\geq \lim\limits_{m\to +\infty} -\sum_{\alpha\in\mathcal{D},|\alpha|<m}P(\alpha)\log_{2}P(\alpha) \\
   &=H(\mathcal{D}).
  \end{aligned}
 \end{equation*}
 Equation \eqref{e14} is proved.
 Finally, form equation \eqref{e14}, we have
 \begin{equation*}
  \begin{aligned}
   H(\mathcal{D})&= \lim\limits_{m\to +\infty} H(\mathcal{D}_{m})       \\
   &= \lim\limits_{m\to +\infty} H(P)\overline{l(\mathcal{D}_{m})}   \\
   &= H(P)\overline{l(\mathcal{D})} .
  \end{aligned}
 \end{equation*}
 \end{enumerate}
 The proof is completed.
\end{proof}
\bibliographystyle{IEEEtran}
\bibliography{IEEEabrv,refs}
\end{document}